\documentclass[conference]{IEEEtran}

\makeatletter
\def\ps@headings{%
\def\@oddhead{}%
\def\@evenhead{}%
\def\@oddfoot{\mbox{}\scriptsize\rightmark \hfil \thepage}%
\def\@evenfoot{\scriptsize\thepage \hfil \leftmark\mbox{}}}
\makeatother

\usepackage[dvips]{color}
\usepackage{epsf}
\usepackage{times}
\usepackage{epsfig}
\usepackage{graphicx}
\usepackage{amsmath}
\usepackage{amssymb}
\usepackage{amsxtra}
\usepackage{here}
\usepackage{rawfonts}
\usepackage{times}
\usepackage{url}
\usepackage{cite}

\newtheorem{corollary}{\bf Corollary}
\newtheorem{theorem}{\bf Theorem}

\newtheorem{lemma}{\bf Lemma}

\newtheorem{definition}{\bf Definition}

\newlength{\aligntop}
\newlength{\alignbot}
\makeatletter
\renewenvironment{align}{%
  \vspace{\aligntop}
  \start@align\@ne\st@rredfalse\m@ne
}{%
  \math@cr \black@\totwidth@
  \egroup
  \ifingather@
    \restorealignstate@
    \egroup
    \nonumber
    \ifnum0=`{\fi\iffalse}\fi
  \else
    $$%
  \fi
  \ignorespacesafterend%
  \vspace{\alignbot}\par\noindent
}

\IEEEoverridecommandlockouts
\begin{document}
\title{Reactive Power Compensation Game under Prospect-Theoretic Framing Effects \vspace{-0.3cm} }

\author{\authorblockN{Yunpeng Wang$^{1, 5}$, Walid Saad$^{2, 4}$, Arif I. Sarwat$^3$, Choong Seon Hong$^4$}\vspace{+0.1cm} \authorblockA{\small
$^1$ Electrical and Computer Engineering Department, University of Miami, Coral Gables, FL, USA, Email: \url{y.wang68@umiami.edu}\\
$^2$ Wireless@VT, Bradley Department of Electrical and Computer Engineering, Virginia Tech, Blacksburg, VA, USA, Email: \url{walids@vt.edu}\\
$^3$ Department of Electrical and Computer Engineering, Florida International University, Miami, FL, USA, Email: \url{asarwat@fiu.edu}\\
$^4$ Department of Computer Engineering, Kyung Hee University, South Korea, Email: \url{cshong@khu.ac.kr}\\
$^5$ Dispatching Control Center, State Grid Beijing Electric Power Company, China\\
\vspace{-0.9cm}
}%
}
\date{}
\maketitle

\begin{abstract}
Reactive power compensation is an important challenge in current and future smart power systems. However, in the context of reactive power compensation, most existing studies assume that customers can assess their compensation value, i.e., Var unit, objectively. In this paper, customers are assumed to make decisions that pertain to reactive power coordination. In consequence, the way in which those customers evaluate the compensation value resulting from their individual decisions will impact the overall grid performance. In particular, a behavioral framework, based on the \emph{framing effect of prospect theory (PT)}, is developed to study the effect of both objective value and subjective evaluation in a reactive power compensation game. For example, such effect allows customers to optimize a subjective value of their utility which essentially frames the objective utility with respect to a reference point. This game enables customers to coordinate the use of their electrical devices to compensate reactive power. For the proposed game, both the objective case using expected utility theory (EUT) and the PT consideration are solved via a learning algorithm that converges to a mixed-strategy Nash equilibrium. In addition, several key properties of this game are derived analytically. Simulation results show that, under PT, customers are likely to make decisions that differ from those predicted by classical models. For instance, using an illustrative two-customer case, we show that a PT customer will increase the conservative strategy (achieving a high power factor) by $29\%$ compared to a conventional customer. Similar insights are also observed for a case with three customers.

\end{abstract}
\vspace{-0cm}

\section{Introduction}
Reactive power compensation, commonly known as Var compensation, aims to improve the efficiency of delivering energy in power systems by reducing transmission losses. This has led to much research that investigates how to control and manage reactive power in a smart grid~\cite{dixon2005reactive}. Delivering energy over power lines will generate active and reactive power, and a suitable reactive power compensation can decrease energy losses and increase the power factor which is defined as the value of the tangent of the angle between active and reactive power~\cite{PS00}. However, due to the aging of the devices (i.e., motors, switches) and the varying energy requirements from end-nodes, smart grid customers may obtain different power factors depending on the same devices that they are previously and currently using. In particular, in the smart grid, the power company can require customers to achieve a given power factor for efficient delivery of AC power~\cite{kouro2010recent}. Recent studies on reactive power compensation have focused on analyzing coordination mechanisms, in which some customers can support extra reactive power on behalf of others, as discussed in ~\cite{blaabjerg2012power, diaz2012review, mohammad2014review}.

Reactive power compensation in the smart grid has been investigated in~\cite{almeida2011optimal, saraswat2013novel, xu2010research, soleymani2013nash}. In particular, reactive power coordination between customers in a local area has been technically introduced at the hardware level, using voltage-source-converter technologies that can both absorb and supply reactive power, as discussed in~\cite{4352074summary, sivachandran2011new, 6495738}. To further explore the coordination between customers, the authors in~\cite{almeida2011optimal} proposed an active-reactive power dispatch procedure to minimize opportunity costs via the use of marginal pricing mechanisms to compensate generators for power provision. The work in~\cite{saraswat2013novel} developed a Pareto-optimization based zonal reactive power market model and a hybrid evolutionary approach was applied in a competitive electricity market. In~\cite{xu2010research}, the authors studied the asynchronous generator system in a wind farm so as to efficiently improve Var compensation between different operating moments of asynchronous generators. The authors in~\cite{soleymani2013nash} allowed the customer to bid reactive power in the energy market as well as maintain the voltage stability margin in an IEEE 39 bus test system. Other related approaches for compensating reactive power are discussed in~\cite{tan2013general, kisacikoglu2010examination, santacana2010getting, bolognani2013distributed}.

The works in~\cite{almeida2011optimal, saraswat2013novel, xu2010research, soleymani2013nash, tan2013general, kisacikoglu2010examination, santacana2010getting, bolognani2013distributed} study reactive power compensation using mathematical tools, such as optimization and game theory. However, most of these existing works assume that customers, as the compensating nodes in the grid, can objectively and precisely assess their power factor compensation, i.e., Var value. However, in practice, customers may have subjective perceptions on how they view such Var values as well as on how other customers compensate reactive power. For example, operating inductive equipment (i.e., motor, relay, speaker, solenoid, transformer and lamp ballast, or even the operation of switched capacitor) will change the tangent relationship between active power and reactive power and then change the transmission losses. This tangent value, or power factor, will impact the active power which, in turn, impacts a customer's electricity bill. To properly study such reactive power compensation one must therefore account for different customers perceptions on the economic gains and losses associated with their bills, which is directly dependent on the active power. In particular, other considerations involve the customers' opinion on the usage of electricity, the reduction of transmission losses, the economic payoffs and the effect of electricity operation and requirement. Thus, when designing power factor compensation and coordination mechanisms, one must take into account such customer-related human factors.

The main contribution of this paper is to propose a new game-theoretic framework to understand how customers can coordinate their reactive power compensations while taking into account their individual subjective perceptions on the economic gains and losses associated with this coordination. We formulate the compensation problem as a static noncooperative game, in which a customer can decide whether or not to act in concert with others, based on reactive power technologies (i.e., install capacitor and voltage support), when their inductive loads change (such as using speakers, cables or motors in a community). In this game, each customer aims to optimize a Var utility that captures the benefits of reaching a high power factor and the associated costs needed to provide reactive power. We allow customers to \emph{subjectively} evaluate their \emph{objective} utility which implies that customers can have different ways to measure the economic benefits that they reap from the power compensation game~\cite{kahneman1979prospect, tversky1981framing, tversky1992advances, laobanPTintro}. Compared to related works on smart grids~\cite{almeida2011optimal, saraswat2013novel, xu2010research, soleymani2013nash, tan2013general, kisacikoglu2010examination, santacana2010getting, bolognani2013distributed, 4352074summary, sivachandran2011new, 6495738}, the contributions of this paper include: \emph{1)} in contrast to conventional game utility, we allow customers to subjectively evaluate their compensation of Var gains and losses and then explore the probability of achieving this compensation; \emph{2)} we design a Var coordination mechanism that encourages customers to efficiently utilize the existing compensating devices and to reach an acceptable power factor required by the grid; and \emph{3)} we develop a distributed algorithm, fictitious play (FP), that is proven to converge to a mixed-strategy Nash equilibrium of the game, thus characterizing the solution under classical game theory and PT. In simulations, our studies show that insightful difference between classical and PT evaluations makes customers change the frequency with which they participate in reactive power compensation, in terms of achieving power factors. Our results also show that zonal compensation can be coordinated via the customers' perception of their Var gains as opposed to their Var losses, which can reduce the overall amount of data collected during reactive power compensation.

The remainder of the paper is organized as follows: Section~\ref{sec:prob} presents the system model and formulates the reactive power compensation as a noncooperative game. In Section~\ref{sec:pt}, we introduce a novel behavioral framework with PT considerations and in Section~\ref{sec:algo} we use FP to solve the game. Simulation results are presented in Section~\ref{sec:sim} while conclusions are drawn in Section \ref{sec:conc}.

\begin{table}[!t]\vspace{-0.2cm}
\small
  \centering
 \caption{
    \vspace*{-0.1cm}Summary of Notations}\vspace{-0.4cm}
\begin{tabular} {c||c}
Symbols & Description\\ \hline\hline 
$p$ & active power\\
$q$ & reactive power\\
$s$ & apparent power\\
$\phi$ & initial power factor (without any compensation)\\
$\widetilde \phi$ & power factor required by grid\\
$q^c$ & a customer's compensating Var value\\
$\widetilde q^c$ & a customer's compensating Var value required by grid\\
$i$ & a customer's index\\
$N$ & the total number of customers\\
$q_i^c(a_i)$ & customer $i$'s Var value using action $a_i$\\
$a_i$ & customer $i$'s (pure) action/strategy\\
$\sigma_i$ & customer $i$'s mixed strategy\\
$u_i$ & the utility of customer $i$'s pure strategy\\
$u_i^0$ & the utility reference point of customer $i$'s pure strategy\\
$U_i$ & customer $i$'s expected utility \\
$U_i^{\text{EUT}}$ & customer $i$'s expected utility under EUT \\
$U_i^{\text{PT}}$ & customer $i$'s expected utility under PT\\
$\tau$ & a penalty factor in Var exchange\\
$\alpha$ & weighting factors to capture gain distortions\\
$\beta$ & weighting factors to capture loss distortions\\
$k$ & aversion parameter to tune losses and gains\\
$m$ & the number of iterations\\
\end{tabular}\label{tab:symbol}\vspace{-0.6cm}
\end{table}

\section{Reactive Power Compensation Model and Game Formulation}\label{sec:prob}

In this section, we first introduce the reactive power compensation model and then, formulate a noncooperative game between the customers. The main notations are listed in Table~\ref{tab:symbol}.

\subsection{Reactive Power Compensation Model}

Consider a smart grid in which each customer has a variable reactive power compensation that depends on each customer's owned equipment~\cite{dixon2005reactive} and~\cite{PS00}. Let $\mathcal{N}$ be the set of all $N$ customers. In general, for reactive power compensation, a customer can install a capacitor or a voltage/current source to reduce the power losses and improve voltage regulation at the load terminals~\cite{dixon2005reactive}. The power company measures the active power and gives customers their optimized power factor (PF). However, existing Var compensation technologies, i.e., using a capacitor, cannot always guarantee reaching a fixed power factor, due to the varying inductive requirement and dynamical operation, i.e., capacitor switching time. Here, we assume that a customer $i \in \mathcal{N}$ requires active power $p_i \in \mathcal{P}$ and causes reactive power $q_i \in \mathcal{Q}$, and thus, its apparent power $s_i \in \mathcal{S}$ is $s_i^2=p_i^2+q_i^2$ and its current PF is $\phi_i \in \Phi$. Each customer will compensate the reactive power and increase its PF to a predefined PF $\widetilde \phi_i \in \widetilde \Phi$, as announced by the power company.

\begin{figure}[!t]
 \begin{center}
 \vspace{-0.1cm}
  \includegraphics[width=8cm]{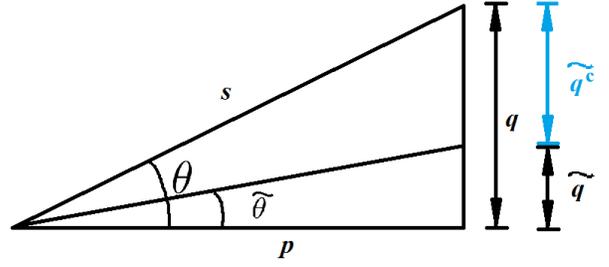}
 \vspace{-0.2cm}
   \caption{\label{fig:fig0} An illustrative example of reactive power compensation.}
\end{center}\vspace{-0.9cm}
\end{figure}

In general, the power factor relates to a phase angle and is defined by the ratio of active power (or real power) $p$ and apparent power $s$ as shown in Fig.~\ref{fig:fig0}. $q$ and $\widetilde q$ are, respectively, the actual reactive power and the required reactive power. Then, after reactive power compensation, customer $i$'s actual compensation is $q_i^c$. Here, we note that the power company will set a desired, compensation requirement/standard. In this regard, we use $\widetilde q_i^c$ to denote this required/standard reactive power compensation for each customer. Due to the delivery of AC power, there exists a capacitor between the power line and ground. In order to effectively deliver the active power and economical consideration, the power company requires customers to reduce their reactive power from $q$ to $\widetilde q$ via a predefined PF. In practice, it is hard to directly measure the PF because of the phase angle between voltage and current. Instead, the company can collect the energy usage of active and apparent power, and, then, send to the customers the tangent value of the angle between active and apparent power, i.e., PF. Hence, we assume that a customer will require a constant active power and varying reactive power, such that its PF can be easily received in the process of Var compensation~\cite{PS00}. From Fig.~\ref{fig:fig0}, we can compute the required Var compensation,
i.e., customer $i$'s required compensation $\widetilde q_i^c$, as follows:
\begin{equation}\label{eq:Var}
\begin{split}
\widetilde q^c_i(\phi_i, \widetilde \phi_i)&=
p_i \cdot \tan \theta_i-p_i\cdot \tan \widetilde \theta_i,\\
&=p_i \cdot \frac{\sqrt{1-\phi_i^2}}{\phi_i}-p_i\cdot \frac{\sqrt{1-\widetilde \phi_i^2}}{\widetilde \phi_i},
\end{split}
\end{equation}
where $\phi_i=\cos \theta_i$ and $\widetilde \phi_i=\cos \widetilde \theta_i$
In practice, it is hard to install new equipment for compensation and the customers can obtain a varying PF due to their over/under compensation. Thus, there might be a need for a Var coordination between customers so as to achieve optimal local compensation. In the studied scenario, it is necessary to devise a mechanism used to understand how the customers compensate reactive power, and how their usage of inductive loads impacts the overall system, in terms of Var benefits and costs. For example, a customer can have some inductive loads, such as speakers in an event or motors for pumping water and, thus, its reactive power requirement increases, as its PF decreases. In such a case, it is difficult to install new capacitors; instead, compensating reactive power from other nodes will be an efficient way to maintain the PF requirement. For example, some customers such as electrical vehicles and reactive power plants, can discharge power to increase PF for the total Var compensation $\sum_{i \in \mathcal{N}} \widetilde q_i^c$. Here, we assume that each customer can obtain/reach a PF via the existing compensating equipment. In this respect, customers will have different power requirements and can achieve a varying PF. Thus, the decisions made by customers will depend on such PF as the operation of the existing devices changes, even if they install new compensating devices. Next, we mainly study the competitive coordination between customers, using their reached PFs, which leads to a game-theoretic setting as discussed next.

\subsection{Noncooperative Game Formulation}

We analyze the operations of compensating reactive power between customers using noncooperative game theory~\cite{GT00}. As previous discussed, the customers compensate reactive power based on their existing equipment. Then, they must make a decision on whether to sell (buy) reactive power to (from) the grid. For the studied model, the compensation value is the reactive power difference between initial PF and the compensated PF reached by customers. For example, the power company allows some customers to buy reactive power from those supplying extra Var compensation, as total Var compensation is satisfied~\cite{Varchange1, Varchange2, Varchange3}. Thus, customer $i$ can compensate reactive power $q_i^c$ and reach its PF $\phi_i^c$, while the power company announces the standard PF $\widetilde \phi_i$.

When coordinating reactive power compensation, customers can \emph{interdependently} determine how much reactive power must be compensated (i.e., Var). We can formulate a static noncooperative game in strategic form $\Xi=[\mathcal{N}, \{\mathcal{A}_i\}_{i\in\mathcal{N}}, \{u_i\}_{i\in\mathcal{N}}]$, that is characterized by three main elements: \emph{1)} the \emph{players} which are the customers in the set $\mathcal{N}$, \emph{2)} the \emph{strategy} or \emph{action} $\mathcal{A}_i:=(\phi_i, 1]$, which represents customer $i$'s achieved PF, and \emph{3)} the \emph{utility function} $u_i$ of any player $i \in \mathcal{N}$, which captures the benefit-cost tradeoffs associated with the different choices. In particular, we hereinafter assume a discrete strategy set. The value of the utility function achieved by a customer $i$ that chooses an action $a_i$ is given by:
\begin{equation}\label{eq:utility}
u_i(a_i, \boldsymbol{a}_{-i})=B_i(a_i, \boldsymbol{a}_{-i})-C_i(a_i, \boldsymbol{a}_{-i}),
\end{equation}
where $\boldsymbol{a}_{-i} = [a_1, a_2, \dots, a_{i-1}, \dots,a_{i+1}, \dots, a_N]$ is the vector of actions of all players other than $i$, $B_i(a_i, \boldsymbol{a}_{-i})$ is the Var benefits customer $i$ obtained if it provides surplus Var to the grid, and $C_i(a_i, \boldsymbol{a}_{-i})$ is the cost in Var coordination. In practice, customer $i$'s action $a_i$ corresponds to deciding on whether to sell or buy reactive power. Then, compared to the required Var compensation $\widetilde q^c$ in (\ref{eq:Var}), customer $i$ will compensate reactive power as follows:
\begin{equation*}
q_i^c(a_i)=q_i^c(\phi_i, a_i)=p_i\cdot \frac{\sqrt{1-\phi_i^2}}{\phi_i}-p_i\cdot \frac{\sqrt{1-a_i^2}}{a_i}.
\end{equation*}
Here, before we study the benefits and costs using reactive power, we first design a Var coordination exchange between customers:
\begin{equation}\label{eq:sharing}
E_i(a_i, \boldsymbol{a}_{-i})=q_i^c(a_i)-\frac{\sum_{j \in \mathcal{N}}q_j^c(a_j)}{N}.
\end{equation}
In particular, $E_i(\cdot)$ is the Var difference between customer $i$ and the average compensation. In (\ref{eq:sharing}), customer $i$'s compensating quantity $q_i^c$ depends on its action $a_i$, i.e., $q_i^c(a_i)$ as the customers' interactions are captured through $q_i^c$. Due to the fact that the total Var compensation $\sum_{i \in \mathcal{N}}q_i^c(a_i)$ is affected by other customers, $E_i(\cdot)$ can have a negative value even if customer $i$'s Var compensation exceeds its standard, i.e., $a_i>\widetilde \phi_i$. Using (\ref{eq:sharing}), the benefit of Var exchange will be:
\begin{equation}\label{eq:Bu}
B(a_i, \boldsymbol{a}_{-i})=
\begin{cases}
E_i(a_i, \boldsymbol{a}_{-i}) &\textrm{if } a_i \ge \widetilde \phi_i \textrm{ and} \sum\limits_{i \in \mathcal{N}}q_i^c \ge \sum\limits_{i \in \mathcal{N}}\widetilde q_i^c, \\
0&\textrm{otherwise.}
\end{cases}
\end{equation}
Moreover, the cost incurred by customer $i$ is
\begin{equation}\label{eq:Cu}
C(a_i, \boldsymbol{a}_{-i})=
\begin{cases}
\tau_i(q_i^c-\widetilde q_i^c)^+ &\textrm{if } a_i \ge \widetilde \phi_i \textrm{ and} \sum\limits_{i \in \mathcal{N}}q_i^c \ge \sum\limits_{i \in \mathcal{N}}\widetilde q_i^c, \\
-E_i(a_i, \boldsymbol{a}_{-i}) &\textrm{if } a_i < \widetilde \phi_i \textrm{ and} \sum\limits_{i \in \mathcal{N}}q_i^c \ge \sum\limits_{i \in \mathcal{N}}\widetilde q_i^c, \\
q_i^c &\textrm{otherwise,}
\end{cases}
\end{equation}
where $(F)^+=\max \{0,F\}$ and $0 \le \tau_i\le 1$ is a penalty factor that weighs the losses of customer $i$ when its Var compensation is greater than the standard compensation $\widetilde q_i^c$.

The utility function in (\ref{eq:utility}) captures both the Var benefit as well as the associated costs of having a high PF. Here, when a customer $i$ requires large reactive power and decreases its PF, i.e., $q_i^c<\widetilde q_i^c$ or $a_i<\widetilde \phi_i$, its benefit in (\ref{eq:Bu}) is zero and its utility is $E_i(a_i, \boldsymbol{a}_{-i})$, while total compensation satisfies power system requirement. In particular, its utility depends on the reactive power coordination, and its Var payment would be given to those who provide extra reactive power compensation. On the other hand, if a customer has a high PF, i.e., $a_i>\widetilde \phi_i$ and provides extra reactive power, it might obtain a benefit due to (\ref{eq:sharing}). By using a high PF, such as by over compensating the PF to $0.95$, one might increase the system voltage~\cite{PS00} and then cause voltage oscillation in the grid. Such an extreme high voltage resulting from the overcompensation will endangers the usage of equipment. Thus, a penalty in (\ref{eq:Cu}) limits the extreme case, if all users pursue high PFs. Without loss generality, we also consider another case upon which the total reactive power compensated by all customers cannot meet the total Var requirement, i.e., $\sum\limits_{i \in \mathcal{N}}\widetilde q_i^c$, and assume that all customers lose their Var values in compensation.

\section{Prospect Theory for Reactive Power Compensation}\label{sec:pt}

In this section, we first study a conventional game solution using expected utility theory to understand how the reactive power compensation game can reach an equilibrium. Then, using prospect theory, we analyze the impact of customer behavior on this game, when customers frame their utility values with respect to a reference point.

\subsection{Reactive Power Compensation under Expected Utility Theory}

Owing to the varying active/reactive power requirement (i.e., charging/discharging, voltage support, and inductive load usage), the PF reached by a customer is not a fixed constant. Also, due to the continuous operation time for the compensation equipment (i.e., switching diodes), the PFs reached after customer compensations are not discrete values but continuous. However, the PF announced by the power company falls within a discrete sample space whose distribution can be specified by a probability mass function. Here, we assume that customers can make probabilistic choices over their discrete strategies and therefore, we are interested in studying the game under \emph{mixed strategies}~\cite{GT00} rather than under \emph{pure, deterministic strategies}. Intuitively, a mixed strategy is a probabilistic choice that captures how frequently a customer will choose a given pure strategy. Such assumption of the mixed, probabilistic choices is motivated by the following factors: \emph{1)} a probability or frequency can represent how often a customer reaches a power factor, and one can better understand how such operations will occur over a large period of time, and \emph{2)} a customer would avoid providing individual power factor compensation information so as to compete with its opponents. In this respect, let $\boldsymbol{\sigma}=[\sigma_1, \sigma_2, \dots, \sigma_N]$ be the vector of all mixed strategies. For customer $i$, its $\sigma_i(a_i) \in \Gamma_i$ is the probability corresponding to its pure strategy $a_i \in \mathcal{A}_i$, where $\Gamma_i$ is the set of mixed strategy available to customer $i$.

In traditional game theory~\cite{GT00}, it is assumed that a player makes rational decisions. Such rational decisions/actions imply that, each player will objectively choose its mixed strategy vector so as to optimize its own utility. Indeed, under the conventional expected utility theory, the utility of each customer is simply the expected value over its mixed strategies and thus, for any player $i \in \mathcal{N}$, its EUT utility is given by:
\begin{equation}\label{eq:multiplayerET}
U_i^{\text{EUT}}( \boldsymbol{\sigma})=\sum_{\boldsymbol{a} \in \mathcal{A}}\bigg(\prod_{j=1}^N \sigma_j(a_j)\bigg) u_i(a_i, \boldsymbol{a}_{-i}),
\end{equation}
where $\boldsymbol{a}$ is a vector of all chosen/played pure strategies and $\mathcal{A}=\mathcal{A}_1 \times \mathcal{A}_2 \times \dots \times \mathcal{A}_N$.

\subsection{Reactive Power Compensation under Prospect Theory}\label{subsec:pt}

Using the game-theoretic formulation in (\ref{eq:multiplayerET}), a player can assess its expected utility, where customers can objectively evaluate Var payoff under EUT. However, because each customer evaluates its economic benefits differently, such a subjective perception will impact the overall results of the reactive power compensation game. For example, for a $1$ kW house usage, the compensation of $100$ Var may be considered by a customer (i.e., require Var from grid), while such $100$ Var might not enable a factory with $100$ kW power requirement to buy Var from grid, due to the small impact on PF. Indeed, due to the different viewpoints on a same Var value, i.e., $100$ Var, a small power customer will prefer to compensate reactive power, while a large power customer might ignore a strategy that small customers choose in compensation. Thus, customers can make subjective evaluations that result in a deviation from the utility in (\ref{eq:utility}). A customer's evaluation can consist of both gains and losses, when it admits a criterion. In particular, the gain (loss) is a positive (negative) value in (\ref{eq:utility}), as the criterion is $0$ for EUT. Therefore, the difference between the subjective evaluation and classical, objective utility in (\ref{eq:utility}) requires one to develop a new framework that can analyze the compensation problem in a smart grid.

To study the customer's behavior, several empirical studies~\cite{kahneman1979prospect, prelec1998probability, PT01, PT02} have analyzed how customer behavior affects a noncooperative game. In a decision-making process, a player can evaluate its utility based on a reference, which represents how this player measures gains and losses with respect to a certain economic reference or framing point (e.g., a level of ``wealth''). To capture how losses loom larger than gains under the perception of customers, one can map/transform the objective utility functions into subjective value functions and, this transformation is the so-called framing effect. In particular, when a customer makes a decision, it will subjectively evaluate its utility, i.e., based on its perception on the Var units of reactive power compensation. Then, over-compensation and under-compensation might lead to specific operational gains or losses. How such gains and losses are evaluated will be given as a new different, customer-dependent utility, i.e., $u_i^{\text{PT}}$. Thus, taking into account a reference point and how benefits and costs are evaluated by each customer, the expression of the utility will be different from that of EUT in (\ref{eq:utility}).

In order to capture the effect of such evaluation, we will use prospect theory~\cite{kahneman1979prospect}. In particular, prospect theory allows framing the utilities based on the following criteria: \emph{1)} Reference point: a player can evaluate its utility using its own individual reference point and such evaluation represents how players act differently via a possibly similar utility value (i.e., a same $\$ 100$ can be evaluated differently by a rich individual compared to a poor individual); \emph{2)} Gain/loss aversion: a player has different attitudes for given a value when it corresponds to a gain as opposed to when it corresponds to a loss; and \emph{3)} Diminishing sensitivity: a player is risk averse in large gain values and risk seeking in small losses. Using these three notions, for each player $i \in \mathcal{N}$, we can review the utility function in (\ref{eq:utility}) and construct a behavioral utility function that can allow the players to evaluate both gains and losses, with the realistic consideration of a utility reference point~\cite{tversky1992advances}:
\begin{equation}\label{eq:uPT}
u_i^{\text{PT}}(\boldsymbol{a})=
\begin{cases}
\bigg(u_i(\boldsymbol{a})-u_i^0(\boldsymbol{a}^0)\bigg)^{\alpha_i} &\textrm{\!\!\!\!\!\!if } u_i(\boldsymbol{a}) \ge u_i^0(\boldsymbol{a}^0), \\
-k_i\bigg(u_i^0(\boldsymbol{a}^0)-u_i(\boldsymbol{a})\bigg)^{\beta_i} &\textrm{\!\!\!\!\!\!otherwise,}
\end{cases}
\end{equation}
where $u_i^0(\boldsymbol{a}^0)=u_i(\boldsymbol{a}^0)$ is the utility reference point based on the strategy vector $\boldsymbol{a}^0$, the weighting factors $\alpha_i, \beta_i \in (0,1]$ respectively capture the gain and loss distortions, and $k_i >0$ is an aversion parameter to tune the impact difference between losses and gains. In this respect, the utility in (\ref{eq:uPT}) is a desired S-shape function and it is concave for gains and convex for losses~\cite{tversky1981framing}. Based on the reference point, smaller $\alpha_i, \beta_i$ will cause a greater distortion in gain and loss magnitudes. Moreover, when $k_i>1$,
player $i$ evaluation will have a stronger impact on its loss than its gain, termed as the case ``loss aversion''~\cite{PTbook01}. Compared to the EUT utility
function in (\ref{eq:multiplayerET}), the expected utility under PT framing is:
\begin{equation}\label{eq:multiplayerPT}
U_i^{\text{PT}}( \boldsymbol{\sigma})=\sum_{\boldsymbol{a} \in \mathcal{A}}\bigg(\prod_{j=1}^N \sigma_j(a_j)\bigg) u_i^{\text{PT}}(a_i, \boldsymbol{a}_{-i}),
\end{equation}
where $\boldsymbol{a}$ is the choosing action vector, as mentioned in (\ref{eq:multiplayerET}), and $u_i^{\text{PT}}$ is the PT pure utility of the action combination.

\section{Game Solution and Proposed Algorithm}\label{sec:algo}

Next, we first show the existence of a mixed NE for the proposed game and then, we prove that using an FP-based algorithm customers can reach a mixed NE in our model.

In (\ref{eq:multiplayerET}) or (\ref{eq:multiplayerPT}), we show the expected utility using the set of mixed strategy over the action set $\mathcal{A}_i$ of each player $i$. The game-theoretic solution for both EUT and PT can be characterized by the concept of a \emph{mixed-strategy Nash equilibrium}:
\begin{definition}
A mixed strategy profile $\boldsymbol{\sigma}^* $ is said to be a mixed strategy Nash equilibrium if, for each player $i \in \mathcal{N}$, we have:
\begin{equation}\label{eq:ne}
U_i(\boldsymbol{\sigma}_i^*,\boldsymbol{\sigma}_{-i}^*) \ge U_i(\boldsymbol{\sigma}_i,\boldsymbol{\sigma}_{-i}^*), \  \forall \boldsymbol{\sigma}_i \in  \Gamma_i.
\end{equation}
Note that the mixed-strategy Nash equilibrium defined in (\ref{eq:ne}) is applicable for both EUT and PT; the difference would be in whether one is using (\ref{eq:multiplayerET}) or (\ref{eq:multiplayerPT}), respectively.
\end{definition}

\begin{lemma}\label{th:exiN}
For the proposed reactive power compensation game, there exists at least one mixed strategy Nash equilibrium for PT.
\end{lemma}

\begin{proof}
In the proposed game, a player will assess the objective utility and follow an EUT strategy using (\ref{eq:utility}) and (\ref{eq:multiplayerET}), while it makes a PT-based decision in (\ref{eq:multiplayerPT}) via estimating the subjective tradeoffs in (\ref{eq:uPT}). Under EUT, it has been shown that there exists at least one mixed strategy Nash equilibrium in a game with a finite number of players, in which each player can choose from finitely many pure strategies. Under PT, both the number of players and the number of their pure strategies do not change. Then, for each pure strategy, the PT utility only reconstructs underlying EUT value; therefore, there exists at least one mixed NE in the PT game, as well as its existence in EUT.
\end{proof}

\begin{corollary}\label{co:pure1}
If no customer reaches the predefined PF in the reactive power compensation game, i.e., $a_i < \widetilde \phi_i, \forall i \in \mathcal{N}$, there exists a \emph{unique}, \emph{pure} Nash equilibrium for both EUT and PT.
\end{corollary}
\begin{proof}
In this case, the total reactive power compensated by all customers does not meet the
total compensation requirement using (\ref{eq:utility}), (\ref{eq:sharing}), (\ref{eq:Bu}) and (\ref{eq:Cu}). In particular, $u_i=-q_i^c(a_i)$ for all customers. For EUT, we have
\begin{equation}
\begin{split}
\frac{\partial u_i}{\partial a_i} =&\frac{\partial u_i}{\partial q_i^c(a_i)} \cdot \frac{\partial q_i^c(a_i)}{\partial a_i} \\
=&-p_i \cdot \frac{1}{\sqrt{1-a_i^2}\cdot a_i^2}<0,
\end{split}
\end{equation}
where $\phi_i<a_i<\widetilde \phi_i$. Thus, player $i$ will follow a dominant strategy\footnote{A strategy is said to be a dominant strategy for a player if it yields the best utility (for that player) no matter what strategies the other players choose.}, i.e., $a_i^{\text{min}}$. Then all EUT customers will choose their dominant strategies as a unique, pure NE. Similarly, for PT
\begin{equation}\label{eq:framvalue}
\frac{\partial u_i^{\text{PT}}}{\partial a_i} =
\begin{cases}
-\alpha_i p_i \cdot \frac{1}{\sqrt{1-a_i^2}\cdot a_i^2}  \cdot (u_i(\boldsymbol{a})-u_0^0(\boldsymbol{a}^0))^{\alpha_i-1} \\
\qquad \qquad \qquad \qquad \qquad \textrm{ if } u_i(\boldsymbol{a}) \ge u_i^0(\boldsymbol{a}^0), \\
-k_i\beta_i p_i \cdot \frac{1}{\sqrt{1-a_i^2}\cdot a_i^2} \cdot (u_0^0(\boldsymbol{a}^0)-u_i(\boldsymbol{a}))^{\beta_i-1}\\ \qquad \qquad \qquad \qquad \qquad \textrm{ otherwise.}
\end{cases}
\end{equation}
Here, both $(u_i(\boldsymbol{a})-u_0^0(\boldsymbol{a}^0))^{\alpha_i-1}$ and $(u_0^0(\boldsymbol{a}^0)-u_i(\boldsymbol{a}))^{\beta_i-1}$ are greater than $0$. Hence, $\frac{\partial u_i^{\text{PT}}}{\partial a_i}<0$ and all PT customers will choose the dominant strategy as a unique, pure NE strategy. In particular, the unique, pure strategy is to choose the minimum PF strategy, i.e., $a_i^{\text{min}}$, in the strategy set.
\end{proof}
\begin{corollary}\label{co:pure2}
If all customers exceed the predefined PF in the reactive power compensation game, i.e., $a_i > \widetilde \phi_i, \forall i \in \mathcal{N}$, and the penalty factor will not be equal to the ratio of all customers minus one (i.e., without customer $i$) to the total number of customers, i.e., $\tau_i \neq \frac{N-1}{N}, \forall i$, there exists a \emph{unique}, \emph{pure} Nash equilibrium for both EUT and PT.
\end{corollary}
\begin{proof}
In this case, the utility of player $i$ is:
\begin{equation}
\begin{split}
u_i=&q_i^c(a_i)-\frac{\sum_{i \in \mathcal{N}}q_i^c(a_i)}{N}-\tau_i\biggl(q_i^c(a_i)-\widetilde q_i^c(\widetilde \phi_i)\biggr)\\
=&(\frac{N-1}{N}-\tau_i)q_i^c(a_i)-\frac{\sum_{l \neq i, l \in \mathcal{N}}q_l^c(a_l)}{N}+\tau_i \widetilde q_i^c(\widetilde \phi_i)
\end{split}
\end{equation}
Under both EUT and PT, the utility derivatives on player $i$'s strategy are given by:
\begin{equation}
\begin{split}
\frac{\partial u_i}{\partial a_i} =&(\frac{N-1}{N}-\tau_i)\cdot p_i \cdot \frac{1}{\sqrt{1-a_i^2}\cdot a_i^2},\\
\frac{\partial u_i^{\text{PT}}}{\partial a_i} =&\begin{cases}
(\frac{N-1}{N}-\tau_i)\cdot \alpha_i p_i \cdot \frac{1}{\sqrt{1-a_i^2}\cdot a_i^2} \cdot (u_i(\boldsymbol{a})-u_0^0(\boldsymbol{a}^0))^{\alpha_i-1} \\
\qquad \qquad \qquad \qquad \qquad \textrm{ if } u_i(\boldsymbol{a}) \ge u_i^0(\boldsymbol{a}^0), \\
(\frac{N-1}{N}-\tau_i)\cdot k_i\beta_i p_i \cdot \frac{1}{\sqrt{1-a_i^2}\cdot a_i^2} \cdot (u_0^0(\boldsymbol{a}^0)-u_i(\boldsymbol{a}))^{\beta_i-1}\\ \qquad \qquad \qquad \qquad \qquad \textrm{ otherwise.}
\end{cases}
\end{split}
\end{equation}
Thus, as $\tau_i \neq \frac{N-1}{N}$, both EUT and PT utilities are monotonic function on $a_i$, and all players will choose their dominant strategies as a unique, pure NE. In particular, \emph{1)} when $\tau_i < \frac{N-1}{N}$, the NE is the maximum PF strategy set, \emph{2)} when $\tau_i > \frac{N-1}{N}$, the NE is the minimum PF strategy set, \emph{3)} when $\tau_i = \frac{N-1}{N}$, the mixed NEs are not unique.

The ratio $\frac{N-1}{N}$ is a value that depends only on the number of customers. It captures how the extra compensation of one customer will be shared by the others. Indeed, based on various conditions related to $\tau$, Corollary 2 shows a specific case in which customers choose the same NE under both EUT and PT.
\end{proof}

Corollary~\ref{co:pure1} and Corollary~\ref{co:pure2} mainly provide the analysis when the total customers' compensation is strictly less/greater than the total standard compensation. Next, we will study a two-customer case, when one does not satisfy its compensation requirement and exactly requires compensation from the other one.
\begin{corollary}\label{co:pure3}
For a two-customer reactive power compensation game, if both customers require different power and exceed the predefined PF $\widetilde \phi_1=\widetilde \phi_2=\widetilde \phi$ using a pair of actions, i.e., $A_i=\{v_1, v_2\}$ $(v_1<v_2)$, $\frac{v_1+v_2}{2}=\widetilde \phi$, then, there exists a \emph{unique}, \emph{mixed} Nash equilibrium for both EUT and PT.
\end{corollary}
\begin{proof}
Without loss of generality, we assume that $p_1<p_2$ in the following proof. When the players exceed the predefined PF using a pair of actions $(v_1, v_2)$, the total compensation must satisfy
\begin{equation}
q_1^c(v_1)+q_2^c(v_2)>q_1^c(\widetilde \phi)+q_2^c(\widetilde \phi).
\end{equation}
To derive this equation, we have
\begin{equation}
p_1V(v_1)+p_2V(v_2)<p_1V(\widetilde \phi)+p_2V(\widetilde \phi),
\end{equation}
where $V(x)=\frac{\sqrt{1-x^2}}{x}$. For another action combination $(v_2, v_1)$, we need to compare $q_1^c(v_2)+q_2^c(v_1)$ and $q_1^c(\widetilde \phi)+q_2^c(\widetilde \phi)$. In particular, $q_1^c(v_2)+q_2^c(v_1)-q_1^c(\widetilde \phi)-q_2^c(\widetilde \phi)=(p_1+p_2)V(\widetilde \phi)-p_1V(v_2)-p_2V(v_1)$. Since $V(x)$ is decreasing and convex in $[0,1]$, we have
\begin{equation}
\begin{split}
\frac{p_1V(v_2)}{p_1+p_2}+\frac{p_2V(v_1)}{p_1+p_2}\ge& V\biggl(\frac{p_1}{p_1+p_2}v_2+\frac{p_2}{p_1+p_2}v_1\biggr)\\
=&V\biggl(\widetilde \phi+\frac{p_1-p_2}{2(p_1+p_2)}(v_2-v_1)\biggr)\\
>&V(\widetilde \phi ).
\end{split}
\end{equation}
Thus, $q_1^c(v_2)+q_2^c(v_1)<q_1^c(\widetilde \phi)+q_2^c(\widetilde \phi)$. This inequality implies that, $u_i(v_2, v_1)=-q_i^c(a_i)$ for both customers and thus, they will lose their reactive power using a pair of actions $(v_2, v_1)$.

\vspace{-0.1cm}
\begin{center}
Reactive power compensation: (Customer $1$, Customer $2$)
\begin{tabular} {c||c||c}
\small
 & Player $2$'s $v_1$ & Player $2$'s $v_2$\\ \hline\hline 
Player $1$'s $v_1$ & $u_1(v_1,v_1)$, $u_2(v_1,v_1)$&$u_1(v_1,v_2)$, $u_2(v_1,v_2)$\\
Player $1$'s $v_2$ &$u_1(v_2,v_1)$, $u_2(v_2,v_1)$ & $u_1(v_2,v_2)$, $u_2(v_2,v_2)$\\
\end{tabular}\vspace{+0cm}
\end{center}

The above table is the utility of the proposed noncooperative matrix game. To compare the utility values in the matrix game, we must first define the notion of \emph{best response}:
\begin{definition}
The \emph{best response} $\text{br}(\boldsymbol{a}_{-i})$ of any storage unit $i \in \mathcal{N}$  to the vector of strategies $\boldsymbol{a}_{-i}$ is a set of strategies for seller $i$ such that:
\begin{align*}
\text{br}(\boldsymbol{a}_{-i})\!=\!\{a_i \in \mathcal{A}_i|U_i(a_i,\boldsymbol{a}_{-i}) \ge U_i(a_i^\prime,\boldsymbol{a}_{-i}),\ \forall a_i^\prime \in \mathcal{A}_i\}.
\end{align*}
\end{definition}
Using the concept of best response, for any customer $i \in \mathcal{N}$, when the other customers' strategies are chosen as given by $\boldsymbol{a}_{-i}$, any best response strategy in $\text{br}(\boldsymbol{a}_{-i})$ is at least as good as any other strategy in $\mathcal{A}_i$. Under EUT, since $u_1(v_1,v_1)>u_1(v_2,v_1)$, customer $1$ will pick the action $v_1$ as customer $2$ chooses $v_1$; for customer $2$, since $u_2(v_1,v_2)>u_2(v_1,v_1)$, it will pick the action $v_2$ as customer $1$ chooses $v_1$. Under PT, because the framing utility in (\ref{eq:uPT}) only changes the absolute difference between PT (pure) utility and EUT (pure) utility, a PT customer does not change its picking strategy as its opponent holds. Thus, there exists a unique, mixed NE under both EUT and PT.

In particular, for EUT, as $\tau$ varies, the proposed game can have three cases: \emph{1)} when $\tau$ is small, we can obtain $u_1(v_2,v_2)>u_1(v_1,v_2)$ and $u_2(v_2,v_2)>u_2(v_2,v_1)$, thus, there is a unique, pure NE $(v_2,v_2)$; \emph{2)} when $\tau$ is large, we can obtain $u_1(v_2,v_2)<u_1(v_1,v_2)$, thus, there is a unique, pure NE $(v_1,v_2)$; and \emph{3)} when $\tau$ is a median value, we can obtain $u_1(v_2,v_2)>u_1(v_1,v_2)$ and $u_2(v_2,v_2)<u_2(v_2,v_1)$, thus, there is a unique, proper mixed NE. For PT, we will have the same conclusion due to the framing utility $u_i^{\text{PT}}$ in (\ref{eq:uPT}).

In a practical system, we have the following scenarios: \emph{1)} when $\tau$ is small, the cost/penalty of providing reactive power to the grid is small and, thus, both customers will seek to compensate reactive power. \emph{2)} When $\tau$ is large, the cost/penalty of providing reactive power is large. However, if the total compensation cannot satisfy the Var requirements (both customers choose a small PF strategy), the customers' compensation action will be penalized. Thus, these two customers will then compensate with each other so as to avoid such a cost/penalty. \emph{3)} When $\tau$ is neither too large nor too small, the cost/penalty might be equal to the compensation of choosing the small PF strategy. Thus, customers will have a mixed strategy.
\end{proof}

To complete such compensation between two customers, as per Corollary~\ref{co:pure3},  the grid operator can announce the PFs and based on wireless technologies, two customers will obtain the PF information. Furthermore, to extend the two-by-two interactions to a general case, we can divide the area of interest into several areas where two neighbors can have a peer-to-peer compensation. \textcolor{black}{For example, consider a scenario with $5$ customers in two areas participate in reactive power compensation. In particular, Area $A$  involves Customer $A1$, $A2$ and $A3$ while Area $B$ involves Customer $B1$ and $B2$. In particular, $P_\text{A1}=70, \phi_\text{A1}=0.81$, $P_\text{A2}=30, \phi_\text{A2}=0.87$, $P_\text{A3}=40, \phi_\text{A3}=0.84$ and $P_\text{B1}=43, \phi_\text{B1}=0.86$, $P_\text{B2}=43, \phi_\text{B2}=0.88$. Then, the power factors in both areas can be obtained by integrating the customers' active powers and factors in each area, i.e., $\phi_A= 0.83$ and $\phi_B= 0.87$. Then, these two areas have a pair of power factors (actions). Thus, for the power company, the customers can be first divided into two areas using a pair of power factor (even if the number of customers in each area is different). }

To solve the compensation game and find an NE using a suitable algorithm, under both EUT and PT, a fictitious play-based algorithm is proposed in Table~\ref{tab:algo}. In this algorithm, the first stage involves a simple initialization, in which each customer translates its action, i.e., the reaching PF, into the Var compensating value. Then, we propose an iterative process based on the fictitious play algorithm~\cite{LEARN07} for solving the game in the second learning stage, under both EUT and PT. Here, the customers will observe others strategies at time $m-1$ so as to update their next strategies at time $m$. In this respect, the customers will update their beliefs about each other's strategies by monitoring their actions. We let $a_i(m)$ be the action taken by player $i$ at time $m$ and $\sigma_i^{a_i}(m),\ a_i\in\mathcal{A}_i, i \in \mathcal{N}$, be the empirical frequency, representing the frequency that player $i$ has chosen strategy/action $a_i$ until time $m$. At any given iteration $m$, the following FP process is used by a player $i$ to update its beliefs:
\begin{align}\label{eq:algo}
\sigma_i^{a_i}(m)=\frac{m-1}{m} \cdot \sigma_i^{a_i}(m-1) + \frac{1}{m}\cdot \mathbf{1}_{\{a_i(m-1)=a_i(m)\}}.
\end{align}

\begin{table}[!]
\centering
\caption{Reactive Power Compensation using Proposed FP}
\begin{tabular}{p{8cm}}
 \hline
      \textbf{Stage 1 - Initialization} \vspace*{.1em}\\
   \hspace*{1em} Customer $i$ chooses a certain initial mixed strategy vector $\boldsymbol{\sigma}_i^{\textrm{init}}$.\vspace*{.1em}\\
   \hspace*{1em} Compute the standard and current compensation of customer $i$, i.e., Var value, using $\phi_i, \widetilde \phi_i, p_i, \boldsymbol{a}_i$, (\ref{eq:Var})-(\ref{eq:Cu}).
\vspace*{.2em}\\
\textbf{Stage 2 - Equilibrium Learning for both EUT and PT,}\vspace*{.2em}\\
\hspace*{1em}\textbf{repeat,}\vspace*{.2em}\\
\hspace*{1em}Each player $i \in \mathcal{N}$ observes the actions of its opponent at time \vspace*{.1em}\\
\hspace*{1em}$(m-1)$:\vspace*{.1em}\\
\hspace*{5em}$\boldsymbol{a}_{-i}(m-1)$; \vspace*{.4em}\\
\hspace*{1em}Compute all expected utilities of each pure strategy, (\ref{eq:multiplayerET})-(\ref{eq:multiplayerPT}), (\ref{eq:FPaction}):\vspace*{.4em}\\
\hspace*{5em}$U_i^{\text{EUT}}( \boldsymbol{\sigma}), U_i^{\text{PT}}(\boldsymbol{\sigma}), a_i(m)$;\vspace*{.4em}\\
\hspace*{1em}Each player $i \in \mathcal{N}$ takes/chooses action $a_i(m)$ as per (\ref{eq:FPaction}): \vspace*{.4em}\\
\hspace*{5em}$a_i(m)=\arg\max {u}_{i}\biggl(a_i,\boldsymbol{\sigma}_{-i}(m-1)\biggr)$;\vspace*{.4em}\\
\hspace*{1em}At time $m$, player $i$'s probability/frequency vector will be changed\vspace*{.1em}\\
\hspace*{1em}as per (\ref{eq:algo}), corresponding to player $i$'s pure actions $\boldsymbol{a}_i$: \vspace*{.4em}\\
\hspace*{5em}$\sigma_i^{a_i}(m)=\frac{m-1}{m} \cdot \sigma_i^{a_i}(m-1) + \frac{1}{m}\cdot \mathbf{1}, a_i\in\mathcal{A}_i$; \vspace*{.4em}\\
\hspace*{1em}\textbf{until} \vspace*{.1em}\\
\hspace*{1em}convergence to a stopping criterion for mixed-strategy NE: \vspace*{.1em}\\
\hspace*{5em}$|\boldsymbol{\sigma}_i(m-1)-\boldsymbol{\sigma}_i(m)|<0.0001$.\vspace*{.4em}\\
\textbf{Stage 3 - Var coordination between customers}\vspace*{.1em}\\
   \hspace*{1em}Customers compensate their reactive power based on the power factor. In a local area, customers will exchange reactive power which will allow them to compensate their PFs.\vspace*{.2em}\\
   \hline
    \end{tabular}\label{tab:algo}\\
\vspace*{-0.4cm}
\end{table}

The strategy chosen at time $m$ is the one that maximizes the expected utility with respect to the updated empirical frequencies. This expected utility would follow (\ref{eq:multiplayerET}) for EUT and (\ref{eq:multiplayerPT}) for PT. Thus, player $i$ can repeatedly choose its strategy as:

\begin{align}\vspace{+0.05cm}
\label{eq:FPaction}
a_i(m)=\arg\max_{a_i\in\mathcal{A}_i}{u}_{i}\biggl(a_i,\boldsymbol{\sigma}_{-i}(m-1)\biggr),
\end{align}\vspace{+0.05cm}

\noindent where the utility here is the expected value obtained by player $i$ with respect to the mixed strategy of its opponents, when player $i$ chooses pure strategy $a_i$. If the chosen strategy $a_i(m)$ is not a singleton, there exists at least one strategy, in which the utility of the strategy is the maximum value in a certain iteration. In particular, if there are more than one strategy that maximizes the utility in (\ref{eq:FPaction}), we will pick the smaller pure strategy, which makes economic sense.

For some specific games, it is well known that FP is guaranteed to converge to a mixed strategy NE~\cite{LEARN07}, as the choosing frequency of players' beliefs converge to a fixed point. However, to our knowledge, such a result has not been extended to PT, as done in the following theorem:

\begin{theorem}\label{th:cov2}
For the proposed reactive power compensation game, the proposed FP-based algorithm is guaranteed to converge to a mixed NE under both EUT and PT, if the choosing frequency of players' beliefs converges in the FP iterative process.
\end{theorem}

\begin{proof}

The convergence of FP to a mixed strategy NE for EUT under the convergence of choosing frequency is a known result as discussed in~\cite{GT00} and~\cite{LEARN07}. For PT, if the choosing frequency converges to a fixed point, this point will be a mixed strategy NE. We prove this case using contradiction as follows.

Suppose that $\{\boldsymbol{\sigma}_k\}$ is a fictitious play process that will converge to a fixed point, i.e., a mixed strategy $\boldsymbol{\sigma}^*$, after $m=n_0$ iterations. By contradiction, we start to assume that the point $\boldsymbol{\sigma}^*=\{\boldsymbol{\sigma}^*_i,\boldsymbol{\sigma}^*_{-i}\}$ is not a mixed strategy NE. Then, \emph{1)} there must exist a strategy $\sigma'_i(a'_i) \in \boldsymbol{\sigma}^*_i$, such that $\sigma_i(a_i)>0, \sigma_i(a_i) \in \boldsymbol{\sigma}^*$ (at least one mixed strategy of player $i$ is not zero) and

\begin{align}\vspace{+0.05cm}
u_i^{\text{PT}}\biggl(a'_i,\boldsymbol{\sigma}^*_{-i}\biggr)>u_i^{\text{PT}}\biggl(a_i,\boldsymbol{\sigma}^*_{-i}\biggr),
\end{align}\vspace{+0.05cm}

\noindent where $u_i(a_i,\boldsymbol{\sigma}^*_{-i})$ is the expected utility with respect to the mixed strategies of the opponents of player $i$, when player $i$ chooses pure strategy $a_i$. Here, we can choose a value $\epsilon$ that satisfies \emph{2)} $0<\epsilon<\frac{1}{2}|u_i^{\text{PT}}(a'_i,\boldsymbol{\sigma}^*_{-i})-u_i^{\text{PT}}(a_i,\boldsymbol{\sigma}^*_{-i})|$ as $\boldsymbol{\sigma}$ converges to $\boldsymbol{\sigma}^*$ at iteration $m=n_0$. Also, \emph{3)} since the FP process decreases as the number of iterations $n$ increases, the utility distance of a pure strategy between two consecutive iterations must be less than $\epsilon$ after a certain iteration $n_0$. For $n \ge n_0$, the FP process can be written as:
\begin{equation}\label{eq:ineq}
\begin{split}
u_i^{\text{PT}}(a_i,\boldsymbol{\sigma}_{-i}^n)=&\sum_{\boldsymbol{a} \in \mathcal{A}} u_i^{\text{PT}}(a_i,\boldsymbol{a}_{-i}^n) \boldsymbol{\sigma}_{-i}^n\\
\le &\sum_{\boldsymbol{a} \in \mathcal{A}} u_i^{\text{PT}}(a_i,\boldsymbol{a}_{-i}^*) \boldsymbol{\sigma}^*_{-i}+\epsilon\\
< &\sum_{\boldsymbol{a} \in \mathcal{A}} u_i^{\text{PT}}(a'_i,\boldsymbol{a}_{-i}^*) \boldsymbol{\sigma}^*_{-i}-\epsilon\\
\le &\sum_{\boldsymbol{a} \in \mathcal{A}} u_i^{\text{PT}}(a'_i,\boldsymbol{a}_{-i}^n) \boldsymbol{\sigma}_{-i}^n\\
=&u_i^{\text{PT}}(a'_i,\boldsymbol{\sigma}_{-i}^n).
\end{split}
\end{equation}

In (\ref{eq:ineq}), we compute the expected utility of pure strategy $a_i$ over the probabilities of all possible cases with respect to the utilities. We obtained the first inequality between two consecutive iterations as in \emph{3)}. We obtained the second inequality using \emph{1)} and \emph{2)}. Then, we obtained the third inequality like the first one, due to \emph{3)}. At last we obtained the expected utility of pure strategy $a'_i$.

Thus, player $i$ would not choose $a_i$ but would rather choose $a'_i$ after the $n$th iteration, mathematically, we will have $\sigma_i(a_i)=0$. Hence, we get $\sigma_i(a_i)=0$ which contradicts the initial assumption that $\sigma_i(a_i)>0$; thus the theorem is shown.
\end{proof}\vspace{+0.1cm}

Following the convergence to a mixed-strategy Nash equilibrium, the last stage in the algorithm of Table~\ref{tab:algo} is how the customers compensate their reactive power in practice and exchange Var between customers. The actual process of Stage 3 is beyond the scope of this paper and will follow economic and real-world contract negotiations.

The algorithm in Table~\ref{tab:algo} shows how customers act in concert with each other for the purpose of reactive power compensation. Such process requires the power company to investigate customers' perception on compensation as captured by the rationality parameters $\alpha$, $\beta$ and $k$ in (\ref{eq:framvalue}). Furthermore, the power company wants to study the relationship between EUT and PT so as to draft a contract with customers. Thus, we next find when the EUT utility is exactly equal to PT result, i.e., the intersection point between EUT and PT.

\begin{theorem}\label{th:kEUTPT}\vspace{+0.05cm}
For the proposed reactive power compensation game, for every customer $i$, there exists a threshold $k_0$, such that, when $k_i<k_0$,
\textcolor{black}{$U_i^{\text{PT}}(\boldsymbol{\sigma}^{\text{PT}*})>U_i^{\text{EUT}}(\boldsymbol{\sigma}^{\text{EUT}*})$}, and when $k_i>k_0$,
\textcolor{black}{$U_i^{\text{PT}}(\boldsymbol{\sigma}^{\text{PT}*})<U_i^{\text{EUT}}(\boldsymbol{\sigma}^{\text{EUT}*})$}.
\end{theorem}\vspace{+0.05cm}

\begin{proof}

In the proposed game, the utility derivative on $k$ can be obtained by (\ref{eq:multiplayerET}) and (\ref{eq:multiplayerPT}). The partial derivative of $U_i$ with respect to $k_i$ depends on the expected utility, while the partial derivative of $u_i$ depends on the utility of a pure strategy. In (\ref{eq:uPT}), the pure PT utility is divided as two cases by the reference point; in (\ref{eq:multiplayerPT}), the expected PT utility can be also viewed as a summation of such two cases. \textcolor{black}{Using the PT utility $u_i^{\text{PT}}$ of a certain pure strategy set in (\ref{eq:uPT}), we can obtain the PT expected utility, as the NE configuration defined in (\ref{eq:multiplayerPT}), }

\vspace{-0.2cm}
\begin{equation*}
\textcolor{black}{\begin{split}
U_i^{\text{PT}}(\boldsymbol{\sigma})=&\sum_{\boldsymbol{a} \in \mathcal{A}}\bigg(\prod_{j=1}^N \sigma_j(a_j)\bigg) u_i^{\text{PT}}(a_i, \boldsymbol{a}_{-i}),\\
=&\sum_{\boldsymbol{a} \in \mathcal{A}, u_i>u_i^0}\bigg(\prod_{j=1}^N \sigma_j(a_j)\bigg) u_i^{\text{PT}}(a_i, \boldsymbol{a}_{-i})\\
&+\sum_{\boldsymbol{a} \in \mathcal{A}, u_i=u_i^0}\bigg(\prod_{j=1}^N \sigma_j(a_j)\bigg) u_i^{\text{PT}}(a_i, \boldsymbol{a}_{-i})\\
&+\sum_{\boldsymbol{a} \in \mathcal{A}, u_i<u_i^0}\bigg(\prod_{j=1}^N \sigma_j(a_j)\bigg) u_i^{\text{PT}}(a_i, \boldsymbol{a}_{-i})\\
=&U_i^{\text{PT}}(\boldsymbol{\sigma})\cdot 1_{u_i>u_i^0}+ U_i^{\text{PT}}(\boldsymbol{\sigma})\cdot 1_{u_i<u_i^0}.
\end{split}}
\end{equation*}

\textcolor{black}{Here, we can get $\sum_{\boldsymbol{a} \in \mathcal{A}, u_i=u_i^0}\bigg(\prod_{j=1}^N \sigma_j(a_j)\bigg) u_i^{\text{PT}}(a_i, \boldsymbol{a}_{-i})=0$. In particular, as $u_i<u_i^0$, $u_i^{\text{PT}}$ can be differentiated by $k_i$ as per (\ref{eq:uPT}). We note that, $U_i^{\text{PT}}$ is a function at $\boldsymbol{\sigma}$, $u_i^0$, $\alpha_i$, $\beta_i$ and $k_i$. In this respect,  $U_i^{\text{PT}}(\boldsymbol{\sigma})=U_i^{\text{PT}}(\boldsymbol{\sigma}, k_i)$}. Thus, to obtain the partial derivative of $U_i$ in $k_i$, we need to consider the partial derivative of both cases in $k_i$:

\begin{equation}\label{eq:k}
\begin{split}
\frac{\partial U_i^{\text{EUT}}\textcolor{black}{(\boldsymbol{\sigma}^*)}}{\partial k_i} =&0,\\
\frac{\partial U_i^{\text{PT}}\textcolor{black}{(\boldsymbol{\sigma}^*, k_i)}}{\partial k_i} =&\frac{\partial U_i^{\text{PT}}\textcolor{black}{(\boldsymbol{\sigma}^*, k_i)}\cdot 1_{u_i>u_i^0}}{\partial k_i}+\frac{\partial U_i^{\text{PT}}\textcolor{black}{(\boldsymbol{\sigma}^*, k_i)}\cdot 1_{u_i<u_i^0}}{\partial k_i}\\
=&0-\sum_{\boldsymbol{a} \in \mathcal{A}, u_i<u_i^0}\bigg(\prod_{j=1}^N \sigma_j(a_j)\bigg)\bigg(u_i^0(\boldsymbol{a}^0)-u_i(\boldsymbol{a})\bigg)^{\beta_i}\\
<&0.
\end{split}
\end{equation}

Here, we note that \emph{1)} the partial derivative of $U_i$ is the expected utility while $u_i$ is the utility of a pure strategy; and \emph{2)} the utility of customer $i$ is a continuous function in $k_i$ while having the discrete action $a_i$. At a mixed NE $\boldsymbol{\sigma}^{\text{EUT}*}$, the objective utility will be a constant value. For PT cases, we can obtain the expected utility via $\boldsymbol{\sigma}^{\text{PT}*}$ and $U_i^{\text{PT}}$ is a strictly decreasing function as $k_i$ increases. Then, $U_i^{\text{PT}}(\boldsymbol{\sigma}^{\text{PT}*})$ and $U_i^{\text{EUT}}(\boldsymbol{\sigma}^{\text{EUT}*})$ will intersect at a point when $k_i=k_0$. In particular, we can compute $k_0$ at the intersected point using the parameters (i.e., $\alpha, \beta, U_i^{\text{EUT}}, u_i^0$). Hence, when $k_i<k_0$, $U_i^{\text{PT}}(\boldsymbol{\sigma}^{\text{PT}*})>U_i^{\text{EUT}}(\boldsymbol{\sigma}^{\text{EUT}*})$, and when $k_i>k_0$, $U_i^{\text{PT}}(\boldsymbol{\sigma}^{\text{PT}*})<U_i^{\text{EUT}}(\boldsymbol{\sigma}^{\text{EUT}*})$. This conclusion implies that, a small (large) $k$ will increase the gain (loss) evaluation, and then increase (decrease) the expected value under PT.
\end{proof}

The PT framing effect is captured via three key parameters: we have three factors, $\alpha_i, \beta_i$ and $k_i$. Compared to $\alpha_i$ and $\beta_i$, the partial derivative of $U_i$ with respect to $k_i$ is more linear. Thus, it is more practical for the power company to control local reactive power compensation via $k_i$ instead of $\alpha_i$ and $\beta_i$. Thus, compared to other factors, the linear property of the aversion parameter $k$ provides a useful approach for the power company to distinguish customers' perception within the proposed compensation game. Theorem~\ref{th:kEUTPT} analyzes the impact of the aversion parameter $k$ instead of the weighting factors $\alpha, \beta$. This theorem investigates the intersection between EUT and PT, and thus, it can be used to compute when the customers' utility is more/less than the standard compensation. In particular, since the derivative of PT utility on $k$ is monotonic in (\ref{eq:k}), customers can have a linear outcome regarding to their perception on compensation gains as opposed to that on compensation losses.

\section{Simulation Results and Analysis}\label{sec:sim}

In this section, we run extensive simulations for understanding how customers' behaviors impact the Var compensation coordination under both EUT and PT game. For simulating the proposed system, we consider a local area consisting of a number of customers equipped with electrical devices to compensate reactive power, i.e., switched capacitors, in which customers' compensation coordination depends on their reaching PF in (\ref{eq:sharing}). To obtain the mixed Nash equilibrium under both EUT and PT, we use the proposed algorithm in Table~\ref{tab:algo}.

\vspace{-0.1cm}
\subsection{Two-customer Case}\label{sec:sim1}

First, we start with the case of two customers. Here, we assume that Customer $1$ and Customer $2$'s initial PF are, $\phi_1=0.77, \phi_2=0.79$, respectively, and their standard PFs are $\widetilde \phi_1=\widetilde \phi_2=0.85$. Also, we assume that the active power $p_1=2$~kW, $p_2=3$~kW with a relative penalty factor $\tau_i=0.7, \forall i$. In particular, both customers choose their strategy from a two-strategy set $\mathcal{A}_i=\{0.8, 0.9\}, \forall i$ as the compensating PF, and their initial mixed strategy sets are $\boldsymbol{\sigma}_1^{\text{init}}=[0.67\ 0.33]^T$ and $\boldsymbol{\sigma}_2^{\text{init}}=[0.2\ 0.8]^T$. In general, if we neglect the impact of the power factor, a bigger active power requirement will lead to a bigger Var payment. In the subsequent simulations, we vary customers' parameters, i.e., $\alpha_i$, $\beta_i$, $k_i$ and $u_i^0$, to gain insights on the proposed Var compensation game under both EUT and PT considerations.

\begin{figure}[!t]
 \begin{center}
 \vspace{-0.3cm}
  \includegraphics[width=8cm]{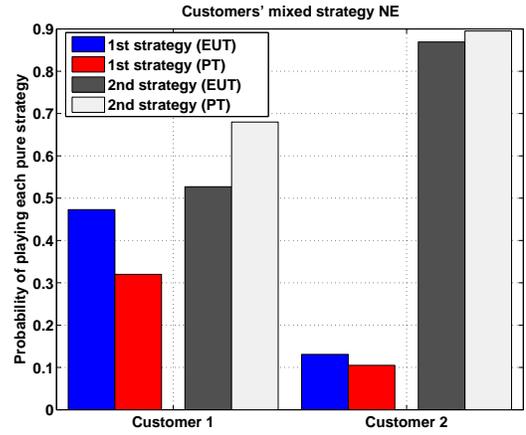}
 \vspace{-0.5cm}
   \caption{\label{fig:bar} Customers' mixed-strategies at the equilibrium for both EUT and PT with $\alpha_i=0.7, \beta_i=0.6, k_i=2, \forall i$.}
\end{center}\vspace{-0.3cm}
\end{figure}

\begin{figure}[!t]
 \begin{center}
 \vspace{-0.3cm}
  \includegraphics[width=8cm]{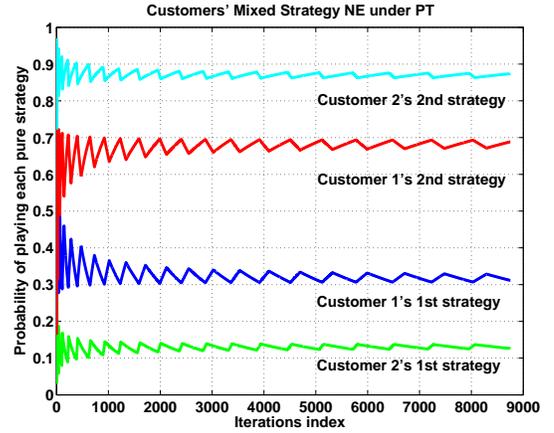}
\vspace{-0.5cm}
   \caption{\label{fig:figiterALLPT} The mixed strategy of playing each pure strategy for both customers as the number of iterations increases. }
\end{center}\vspace{-0.9cm}
\end{figure}

Fig.~\ref{fig:bar} shows the resulting mixed strategies at both the EUT and PT equilibria reached via fictitious play. In this figure, we choose $\alpha_i=0.7$, $\beta_i=0.6$, $k_i=2, \forall i$, and the reference point $u_i^0(\boldsymbol{a}^0)=u_i^0(\widetilde \phi_1, \widetilde \phi_2)$ for two PT customers. In particular, the reference point is chosen to coincide with the case in which customers compensate their reactive power with respect to a standard PF $\widetilde \phi_1=\widetilde \phi_2=0.85$ that is conveyed to customers by the power company, such that, $u_1^0= -0.0256, u_2^0=0.0256$ in (\ref{eq:uPT}). From Fig.~\ref{fig:bar}, we can first see that the mixed strategies of both customers are different between PT and EUT. Under PT, both customers are more likely to choose a high PF compensation action, i.e., $a_1=a_2=0.9$. When a customer chooses a low (high) PF strategy, its reactive power compensation goes below (exceeds) the standard PF compensation of the grid ($\widetilde \phi=0.85$). Thus, under PT, customers evaluate their payoff based on the observation of the standard PF compensation (i.e., $u_i^0(\boldsymbol{a}^0)$), and this will make them avoid taking a low PF action. Indeed, since $\alpha_i>\beta_i$ and $k_i>1,\forall i$, PT gains (losses) will decrease (increase) in (\ref{eq:uPT}) and then, both customers tend to provide extra Var to the grid, instead of risking prospective losses if they use a low PF value. Also, from Fig.~\ref{fig:bar}, we can see the difference in Customer $1$'s strategies chosen in EUT versus PT is larger than that of Customer $2$. In this case, Customer $2$'s active power is larger than Customer $1$. Then, the framing effect on Customer $2$ is smaller than Customer $1$, due to the concavity of gains (i.e., using the large strategy) and the convexity of losses (i.e, using the small strategy) in (\ref{eq:uPT}). Thus, Customer $2$'s EUT strategy would lightly increase via PT considerations.

In Fig.~\ref{fig:figiterALLPT}, we show the values of the PT mixed strategies of both customers (corresponding to Fig.~\ref{fig:bar}), as the number of iterations increases. Here, the proposed algorithm in (\ref{eq:algo}) clearly converges to a mixed NE. The mixed strategy increases to its maximum quickly during the first iteration in (\ref{eq:algo}) and then it decreases as the number of iterations increases. The convergence criterion here is that the difference between two consecutive iterations is small enough in the proposed two-customer game. From Fig.~\ref{fig:figiterALLPT}, the difference between the last two iterations is less than $10^{-4}$ in 1 sec (real time, by using a machine with a 2.2GHz processor and $3$GB RAM). In practice, the company can set a suitable stop criterion for balancing the required communication delay and the convergence time of reactive power compensation.

\begin{figure}[!t]
 \begin{center}
 \vspace{-0.3cm}
  \includegraphics[width=8cm]{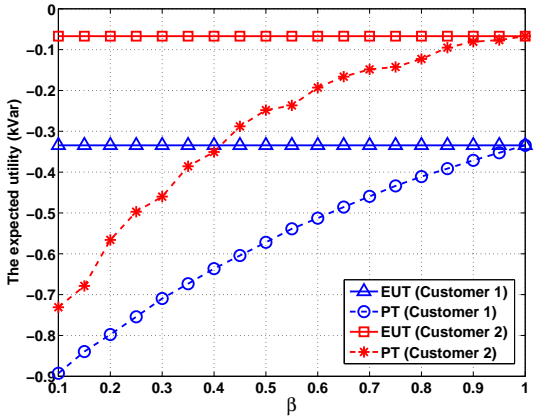}
 \vspace{-0.5cm}
   \caption{\label{fig:beta} The Var expected utility under both PT and EUT as the loss distortion parameter $\beta$ varies.}
\end{center}\vspace{-0.9cm}
\end{figure}

In Fig.~\ref{fig:beta}, we study the costs of compensating reactive power in kVar as the loss distortion parameter $\beta_1=\beta_2=\beta$ increases. In order to singly observe the impact of the loss distortion, we hold $\alpha_1=\alpha_2=1, k_1=k_2=1$ to cancel the gain distortion and aversion effect in (\ref{eq:uPT}). Also, we assume $u_i^0=0, i=\{1, 2\}$ to neglect the impact of reference point. Here the expected utility for both customers is the costs/payment of Var compensation, which is a negative value in (\ref{eq:utility}). In this figure, we can see that the expected utility increases as $\beta$ increases, implying that large $\beta$ decreases PT costs in Var compensation. A small $\beta$ increases the PT loss and, for the proposed prospect model, PT customers will have much cost if they increasingly evaluate PT losing distortion in (\ref{eq:utility}). In particular, when $\beta=1$, the PT utility is equal to the EUT utility. From this figure, we can also see that a same losing parameter $\beta$ leads to different impacts on customers. For example, Customer $2$'s difference between PT and EUT is greater than $0.6$ while that of Customer $1$ is less than $0.6$, when $\beta=0.1$. This is because Customer $2$ requires more active power than Customer $1$.

\begin{figure}[!t]
 \begin{center}
 \vspace{-0.3cm}
  \includegraphics[width=8cm]{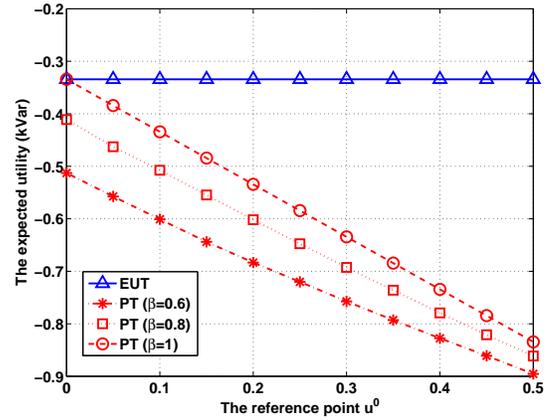}
 \vspace{-0.5cm}
   \caption{\label{fig:refu} The Var expected utility under PT and EUT as the reference point $\boldsymbol{u}^0$ varies.}
\end{center}\vspace{-0.9cm}
\end{figure}

Fig.~\ref{fig:refu} shows the expected Var value under both EUT and PT as the reference point $\boldsymbol{u}^0$ varies. For this scenario, we maintain $\alpha_1=\alpha_2=1, k_1=k_2=1$ to eliminate the impact of gain distortion. Also, we assume both customers have an equal reference and a same loss distortion, i.e., $u_1^0=u_2^0$ and $\beta_1=\beta_2$. First, we can see that the expected Var cost will increase as the reference increases. Because the reference point in (\ref{eq:uPT}) is subtracted from the EUT utility, customer evaluation would depend on a referent level. In essence, a large (active) power requirement leads to more payments in practice, compared to the EUT case. Also, Fig.~\ref{fig:refu} shows that the distortion parameter $\beta$ will have different impacts on the PT utility as well as the references $\boldsymbol{u}^0$. For example, when the reference is a small value, i.e., $u^0=0$, the cost difference between $\beta=1$ and $\beta=0.6$ is around $0.2$, while the difference is less than $0.1$ as $u^0=0.5$. This shows how the distortion parameter impacts the PT utility, which incorporates both the reference and EUT utility in (\ref{eq:uPT}).

\vspace{-0cm}

\subsection{System with More than Three Customers}\label{sec:sim2}\vspace{-0cm}

In Fig.~\ref{fig:player3bar}, we show all mixed strategies in a three-customer game. We choose $\alpha=0.7$, $\beta=0.6$, $k=2, \forall i$ and $\mathcal{A}_i=\{0.8, 0.82, 0.84, 0.86, 0.88, 0.9\}$ for all customers, in which the PT reference point is $u_i^0(\boldsymbol{a}^0)=u_i^0(\widetilde \phi_1, \widetilde \phi_2, \widetilde \phi_3)$ and standard PF is $0.85$. In particular, the active power requirement vector and the initial PF are randomly chosen, respectively, as $\boldsymbol{p}=[2.4\ 4.1\ 3]^T$ and $\boldsymbol{\phi}=[0.77\ 0.78\ 0.77]^T$. Here, we can see that the PT mixed strategies are different from EUT results. Accounting for the PT framing effect in (\ref{eq:uPT}), the reference point $\boldsymbol{u}^0=[-0.1241\ 0.1229\ 0.0012]^T$ and the distortion parameters allow customers to evaluate more on the losses ($\beta, k$) than the gains ($\alpha$). Thus, all customers would want to increase their high PF strategy due to the fact that they observe a prospective losing tendency in practice. Moreover, compared to Fig.~\ref{fig:bar}, we can see that the framing effect in (\ref{eq:uPT}) can change a pure strategy under PT to a more mixed strategy, even change the pure strategy to another pure strategy.

\begin{figure}[!t]
 \begin{center}
 \vspace{-0.3cm}
  \includegraphics[width=8cm]{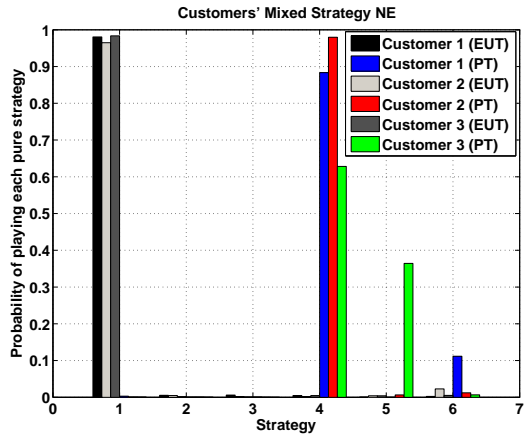}
 \vspace{-0.5cm}
   \caption{\label{fig:player3bar} Customers' mixed-strategies at the equilibrium for both EUT and PT in a multi-customer game.}
\end{center}\vspace{-0.3cm}
\end{figure}

\begin{figure}[!t]
 \begin{center}
 \vspace{-0.3cm}
  \includegraphics[width=8cm]{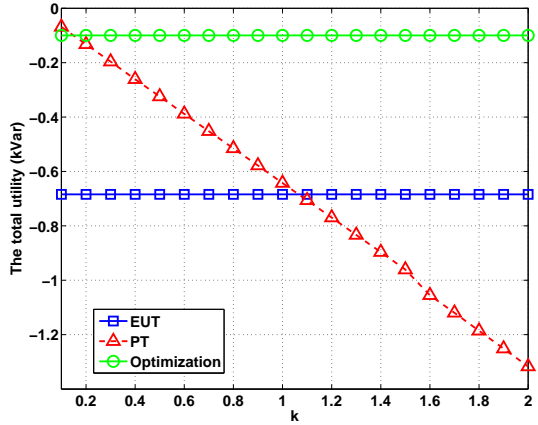}
 \vspace{-0.5cm}
   \caption{\label{fig:fig3total} Customers' total utility under both EUT and PT, as the framing sensitive varies.}
\end{center}\vspace{-0.9cm}
\end{figure}

Using the same parameters as in Fig.~\ref{fig:player3bar}, Fig.~\ref{fig:fig3total} shows the total utility of all customers for the proposed multi-customer game, as the aversion parameter $k_1=k_2=k$ varies. In this figure, we can first see that PT utility decreases as $k$ increases. In (\ref{eq:uPT}), the aversion parameter $k$ can capture customers' perception on evaluation, i.e., gains and losses in practice. This result corresponds to Theorem~\ref{th:kEUTPT}, such that if customers increase evaluation on the gains compared to losses (i.e., $k<1$), their total PT utility will be greater than that of EUT, and vice versa. Indeed, we can see that EUT and PT utilities intersect at a point, i.e., $k=1.04$. The intersection point is not exactly equal to $1$ due to the fact that the gain distortion parameter $\alpha$ is greater than $\beta$ which implies that customers have a smaller gaining distortion than losing distortion. Thus, there needs to be a high aversion parameter (i.e., $k>1$) to balance the distortion difference between gain and loss. Second, compared the EUT and PT results, we show the total utility under optimization (green line). The optimization solution is the total maximum utility of a pure strategy, while EUT and PT show the expected utility over all strategies. From this figure, we can see that the performance of EUT is always less below that of the centralized optimization approach when the power company seeks to maximize the total utility. Although the total optimal utility is greater than the result under EUT or PT, the optimal solution cannot capture customer behavior due to their independence. Thus, we can use to counter the customers' behaviors and design a decentralized, customer-aware optimal solution for reactive power compensation. Last, this figure shows how reactive power will be compensated via zonal/local customers' coordination. For example, compared to the other distortion parameters (i.e., $\alpha, \beta$) that were evaluated in Fig.~\ref{fig:beta}, we can see a more linear curve as the aversion parameter $k$ varies in Fig.~\ref{fig:fig3total}. For example, to collect the PT behavior, the power company needs to investigate how a PT customer may frame its objective gains via  $\alpha$, and how this customer may frame its objective loss via $\beta$. In this case, the power company needs to find two types of information. For the aversion parameter $k$, the power company can directly collect the information how a customer views the PT gains as oppose to the PT losses. This implies that, reactive power compensation can be coordinated via the customers' perception of operational gains as opposed to losses, thus reducing the amount of collected data.

\begin{table}[!t]\vspace{+0.3cm}
\scriptsize
  \centering
  \caption{\vspace*{-0em} All Mixed NE Strategy of $7$ Customers under both EUT and PT, ($\tau_i=\frac{6}{7}, \forall i$)}\vspace*{-0.3cm}
\begin{tabular}{|c|c|c|}
\hline
 & EUT & PT  \\ [0.5ex]\hline
$1$ & $[0.620\ 0.374\ 0.006]^T$ & $[0.333\ 0.666\ 0.001]^T$ \\\hline
$2$ & $[0.633\ 0.363\ 0.005]^T$ & $[0.313\ 0.687\ 0.001]^T$ \\\hline
$3$ & $[0.454\ 0.544\ 0.002]^T$ & $[0.319\ 0.680\ 0.001]^T$ \\\hline
$4$ & $[0.481\ 0.517\ 0.003]^T$ & $[0.286\ 0.714\ 0.001]^T$\\\hline
$5$ & $[0.461\ 0.538\ 0.001]^T$ & $[0.319\ 0.685\ 0.002]^T$\\\hline
$6$ & $[0.514\ 0.482\ 0.005]^T$ & $[0.289\ 0.709\ 0.002]^T$\\\hline
$7$ & $[0.529\ 0.468\ 0.003]^T$ & $[0.303\ 0.696\ 0.001]^T$\\\hline
\end{tabular}
\label{tab:multi}\vspace{-0cm}
\end{table}

Table~\ref{tab:multi} shows all mixed strategies of a system with $7$ customers under both EUT and PT. In this case, we assume that the active power of all customer is $2.4$kW with a three-strategy set, i.e., $p_i=2.4, \mathcal{A}_i=\{0.86, 0.87, 0.88\}, \forall i$. Their standard PFs and initial mixed strategy sets are $\widetilde \phi_i=0.85, \boldsymbol{\sigma}_i^{\text{init}}=[0.33\ 0.33\ 0.33]^T, \forall i$, respectively. In line with Corollary~\ref{co:pure2}, we first set $\tau_i=0.5$. Since $\tau_i=0.5<\frac{6}{7}$, all customers' probabilities on the maximum (pure) strategy are around $1$ and, all customers will choose $0.88$ as their pure compensation strategy under both EUT and PT. Similarly, when $\tau_i=0.9>\frac{6}{7}$, all customers will choose $0.86$ as their compensation pure strategy because all customers' probabilities on the minimum (pure) strategy approach to $1$ under both EUT and PT. Furthermore, to study the difference between EUT and PT, we set $\tau_i=\frac{6}{7}$ so as to guarantee a more ``mixed'' case. In Table~\ref{tab:multi}, the mixed strategies of all PT customers are not the same as the EUT strategies. The difference between EUT and PT mainly pertains to their initial strategies (i.e., $\phi_1=0.79, \phi_5=0.77$) and their participating order in the compensation game (i.e., we use a sequential algorithm in (\ref{eq:algo}). Last but not least, as the number of customers increases, the complexity of finding an NE via FP can increase. For example, a game with each customer having $3$ strategy will have $3^7=2187$ combinations. Increasing the number of customer to $10$ will have $69049$ combinations, which can be too complex to solve. To solve games with number of customers, we can divide the system into multiple, smaller areas and then, within each area applying the proposed scheme to obtain an ``area'' NE. The ``area'' NE can be considered as a player in a large number of customer game.

\begin{figure}[!t]
 \begin{center}
 \vspace{-0cm}
  \includegraphics[width=8cm]{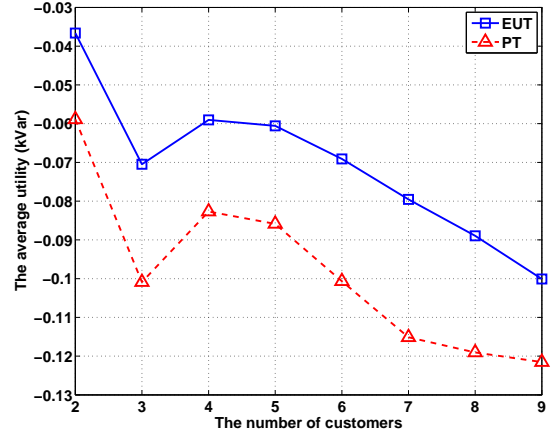}
 \vspace{-0.5cm}
\caption {\label{fig:UtilityVnum} The average utility at the equilibrium for both EUT and PT as $N$ increases.}
\vspace{-0.9cm}
\end{center}
\end{figure}

In Fig.~\ref{fig:UtilityVnum}, we can see that the average utility decreases as the number of customers varies. By assuming $p_i=2.4, \widetilde \phi_i=0.85, \sigma_i^{\text{init}}=1/N,  \mathcal{A}_i=\{0.86, 0.87\}, \forall i$, an increasing number of customers will lead to more interactions as we assume the benefit of Var exchange refers to all customers in (\ref{eq:Bu}). Under the strategy set $\mathcal{A}_i=\{0.86, 0.87\}$, the cost of each customer (i.e., $\tau_i(q_i^c-\widetilde q_i^c)^+$) in (\ref{eq:Cu}) is constant. For a small system ($N\le 3$), the customers might have significantly larger benefits to share than in the four-customer case. When $N\ge 4$, the benefits shared by all customers might decrease while the cost of each customer remains constant. Then, for the proposed model, thus, the shared benefits rely on both the number of customers and their initial active power and, the shared benefits will not always be decreasing. Thus, due to the difference of customers' initial points, at $N=3$, the utility curve has a non-monotonic.

\section{Conclusions}\label{sec:conc}
In this paper, we have introduced a novel game-theoretic approach for modeling the reactive power compensation between local customers via Var coordination. We have formulated the Var compensation process as a noncooperative game between customers, in which customers have subjective perception on their economic losses and gains. Using the framework of prospect theory, we have modeled such perceptions and analyzed their impact on the system. To solve the proposed game, we have proposed a fictitious play-based algorithm that is shown to converge to an equilibrium point under a PT scenario. Simulation results have shown that the use of prospect-theoretic considerations can provide insightful information on the behaviors of customers engaged in reactive power compensation.

\vspace*{+0.1cm}
\def\baselinestretch{0.84}
\bibliographystyle{IEEEtran}
\vspace*{+0.1cm}
\bibliography{references}

\end{document}